\documentclass{article}

\usepackage{microtype}
\usepackage{graphicx}
\usepackage{subfigure}
\usepackage{subcaption}
\usepackage{booktabs} 

\usepackage{hyperref}

\usepackage[accepted]{icml2025}

\usepackage{amsmath}
\usepackage{amssymb}
\usepackage{mathtools}
\usepackage{amsthm}

\usepackage[capitalize,noabbrev]{cleveref}

\usepackage{caption}
\usepackage{pgfplots}
\usepackage{bbm}

\newtheorem{thm}{Theorem}
\newtheorem{cor}[thm]{Corollary}

\newtheorem{prop}[thm]{Proposition}

\theoremstyle{definition}
\newtheorem{df}{Definition}
\newtheorem{eg}{Example}

\theoremstyle{remark}
\newtheorem{rmk}{Remark}
\allowdisplaybreaks

\newcommand{\NEG}{\mathrm{pNE}(G)}
\newcommand{\R}{\mathbb{R}}
\renewcommand{\t}{\mathbf{t}}
\newcommand{\x}{\mathbf{x}}
\newcommand{\X}{\mathcal{X}}
\newcommand{\y}{\mathbf{y}}
\newcommand{\z}{\mathbf{z}}
\newcommand{\argmin}{\mathop{\mathrm{argmin}}}
\newcommand{\argmax}{\mathop{\mathrm{argmax}}}

\allowdisplaybreaks

\theoremstyle{plain}

\theoremstyle{definition}

\theoremstyle{remark}

\usepackage[textsize=tiny]{todonotes}

\icmltitlerunning{The Battling Influencers Game}

\begin{document}

\title{The Battling Influencers Game: Nash Equilibria Structure of a Potential Game and Implications to Value Alignment}
\author{Young Wu, Yancheng Zhu, Jin-Yi Cai, Xiaojin Zhu}
\date{Department of Computer Sciences, University of Wisconsin--Madison}
\maketitle

\begin{abstract}
When multiple influencers attempt to compete for a receiver's attention, their influencing strategies must account for the presence of one another.  We introduce the Battling Influencers Game (BIG), a multi-player simultaneous-move general-sum game, to provide a game-theoretic characterization of this social phenomenon.  We prove that BIG is a potential game, that it has either one or an infinite number of pure Nash equilibria (NEs), and these pure NEs can be found by convex optimization. Interestingly, we also prove that at any pure NE, all (except at most one) influencers must exaggerate their actions to the maximum extent.  In other words, it is rational for the influencers to be non-truthful and extreme because they anticipate other influencers to cancel out part of their influence. We discuss the implications of BIG to value alignment.
\end{abstract}

\section{Introduction} 

Life is full of agents who want to influence others:
Food truck vendors entice us with BBQ samples;
Social media influencers review selective pickleball brands to persuade us;
Editors publish op-eds to sway public opinions.
When multiple influencers with conflicting interests battle for our attention, intuitively they would be strategic and adjust their actions to account for the presence of one another in order to be effective. 

This paper presents a game theoretic definition of the \emph{Battling Influencers Game (BIG)}.  
We model the influencers as players in a multi-player simultaneous-move general-sum game.
Our main technical result is that BIG is a potential game with special pure Nash Equilibria structures.
Consequently, we can predict how rational influencers would adjust their strategies in the battle: \textbf{exaggeration is inevitable}. 
This prediction may shed new computational light on the genesis of misinformation.

As a use case, BIG can be applied to the AI value alignment problem.
The receivers of the influence were traditionally people, but can extend to AI value alignment algorithms.
However, unlike in standard machine learning, our focus is not on the value alignment algorithm itself.
Instead, BIG predicts how battling alignment-data providers could be motivated to intentionally produce training data that do not truthfully reflect their values.
While out of scope for the current paper, our insight can help design future value alignment algorithms to remove such incentives.

\section{Related Work} 

Our work provides a game theoretic model of the numerical example and informal theorem in section $5$ of ~\cite{park2024rlhf}, in particular, we also assume strategic data providers (which we call influencers) to large language models (which we call receivers), and our model leads to results consistent with their example where the data providers untruthfully report their opinions. In addition, we prove the existence of pure strategy Nash equilibria of this class of games, and we show the property that almost all influencers maximally exaggerate their preferences in every equilibrium. Our work is also closely related to ~\cite{hao2024online}, which models the interaction between multiple influencers by a dynamic Bayesian game. They described the phenomenon of strategic extreme exaggeration, which is also discussed in ~\cite{sun2024mechanism}, ~\cite{soumalias2024truthful}, ~\cite{conitzer2024social}, and ~\cite{roughgarden2017online} for various applications, but they do not explicitly compute the equilibria of the original game or quantify the amount of exaggeration. In comparison, we use a static game with known influencer types and we are able to provide better characterizations of the set of equilibria of the game.

Value alignment aims to make language models produce outputs that are more aligned with human values. Existing training frameworks tailored for this purpose, such as ~\cite{ouyang2022training} and ~\cite{rafailov2024direct}, collect preference data from humans and train a large language model to follow users' intent. However, research in this direction mostly focuses on the algorithmic aspects of value alignment and does not emphasize the heterogeneity of human values. There has been work that studies how to make LLMs align with diverse preferences of different demographic groups ~\cite{bakker2022fine, chen2024pal}, but they do not have a rigorous game theoretic foundation that characterizes strategic behaviors of preference data providers. There is also recent work by ~\cite{munos2023nash}, ~\cite{swamy2024minimaximalist}, and ~\cite{rosset2024direct} on learning a pairwise or general preference model, where they used the term Nash learning or Direct Nash Optimization. The players in their zero-sum games are not strategic data providers, and they use the Nash equilibrium solution concept mainly as an optimization technique to solve their minimax problem.



\section{Problem Definition} 

The Battling Influencer Game (BIG) is an $n$-player simultaneous-move game.
The players are the $n$ influencers.
The players have a common continuous action space $\X \subset \R^d$ which is compact and convex.
Let $\x_i \in \X$ be the action of the $i$th player for $i \in [n] \coloneqq \{1, 2, ..., n\}$.
For example, $\x_i$ could be the embedding vector of the text corpus that influencer $i$ produces.
A joint action, or pure strategy profile, $\x=(\x_1, \ldots, \x_n)$ denotes the simultaneous action choice of all influencers.
As in standard literature, we also write $\x=(\x_i, \x_{-i})$ when we want to emphasize player $i$.

For the narrative, we posit a receiver whom the influencers want to influence.  The receiver is not a strategic agent and not a player of the game.
Like an impressionable person, the receiver aggregates to various degrees the inputs it receives from all influencers. 
In this paper we consider affine receivers of the form
\begin{equation}
\label{eq:receiver}
\hat\x := w_0 \x_0 + \sum_{i=1}^n w_i \x_i,
\end{equation}
where $w_i, i \in [n]$ is a real-valued (not necessarily normalized) weight that signifies how much influence influencer $i$ has on the receiver.
For example, a company which can afford to buy more ads has a larger $w_i$ compared to a company with a smaller budget.
$\x_0 \in \X$ is a bias term that, together with $w_0\in \R$, denotes a fixed, constant ``background'' influence that is beyond the control of the $n$ influencers.
The receiver~\eqref{eq:receiver} is common knowledge to all players.

The $n$ influencers each has a target $\t_i \in \R^d$ (i.e. the target is not restricted to $\X$).
For example, $\t_i$ could be the embedding vector of the published party manifesto of political party $i$.
The goal of influencer $i$ is to drive the receiver's $\hat\x$ close to $\t_i$.
This goal is reflected in the loss (negative utility) function of influencer $i$. 
For concreteness, here we consider squared 2-norm as the loss:
\begin{equation}
\label{eq:ell}
\ell_i(\x) := \|\hat\x - \t_i\|_2^2 = \|w_0 \x_0 + \sum_{i=1}^n w_i \x_i - \t_i\|_2^2.
\end{equation}
(See Section~\ref{sec:cosine} for an alternative using cosine similarity.) 
As rational agents, the influencers want to selfishly minimize their own $\ell_i(\x)$.
The fact that the receiver's $\hat\x$ is defined by the joint action $\x$ couples the influencers together in a general-sum game.
It is due to the possible differences in $\t_1, \ldots, \t_n$ that the influencers battle one another.

The above narrative can be abstracted into the following formal definition of BIG, where the receiver becomes implicit:
\begin{df} [Battling Influencer Game (BIG)] \label{df:game}
The Battling Influencer Game is an $n$-player general-sum game $G = \left(n, \X, \{\ell_i\}_{i=1}^{n}\right)$. 
where $n$ is the number of influencers, $\X\subset \R^d$ is a compact and convex action space, and $\ell_i: \X^n \mapsto \R$ taking the form of equation~\eqref{eq:ell} is player $i$'s loss function.
The parameters $\x_0, \{w_i\}_{i=0}^{n}, \{\t_i\}_{i=1}^{n}$ of the loss functions are common knowledge to all players.
\end{df}

In the rest of the paper, we are interested in finding the pure strategy Nash equilibria (NEs) of the game $G$.
We then characterize properties of these NEs, interpreting them in the context of influencers.

\section{Pure NEs of BIG and Their Properties}

\begin{df}  \label{df:ne} 
A pure strategy Nash equilibrium of the game $G = \left(n, \X, \left\{\ell_i\right\}_{i=1}^{n}\right)$ is a strategy profile $\x\in \X^n$ satisfying,
\begin{align}
\ell_{i}\left(\x_{i}, \x_{-i}\right) &\le \ell_{i}\left(\y, \x_{-i}\right), \forall\; \y \in \X, i \in [n] .
\end{align}\end{df}
A mixed strategy Nash equilibrium is a pure strategy Nash equilibrium of the mixed extension $G' = \left(n, \Delta \X, \left\{\ell_i\right\}_{i=1}^{n}\right)$ where the set of actions for each player is a distribution (called a mixed strategy) over the original action set $\X$, that is, a mixed strategy profile $s_{1:n}$ with $s_{i} \in \Delta \X$ satisfying,
\begin{align}
\mathbb{E}\left[\ell_{i}\left(s_{i}, s_{-i}\right)\right] &\leq \mathbb{E}\left[\ell_{i}\left(s', s_{-i}\right)\right], \forall\; s' \in \Delta \X, i \in [n] .
\end{align}
In general, for finite games, for example, when $\X$ is finite, there exists at least one mixed strategy Nash equilibrium, but computing the Nash equilibrium is PPAD-complete (Polynomial Parity Arguments on Directed graphs). When $\X$ is not finite, which is usually the case for our BIG problem, a mixed strategy Nash equilibrium is not guaranteed to exist.

Potential games are games with a special structure and allow strong results on pure Nash equilibria.
We will show BIG is a potential game.
The difficulty is in finding the potential function.
The following theorem provides a constructive proof.

\begin{thm} [Potential Game] \label{thm:ctp} 
The Battling Influencers Game $G$ is a potential game with the potential function
\begin{equation}
\phi(\x)=\phi(\x_1,\ldots,\x_n) := \left\| \sum_{i=0}^n w_i \x_i \right\|_2^2 - 2 \sum_{i=1}^n w_i \t_i^\top \x_i.
\end{equation}
\end{thm}
\begin{proof}  \label{proof:cptpf} 
We need to show if any player $i$ deviates from action $\x_i$ to any action $\y\in\X$,
we have
$\ell_i(\x_i,\x_{-i}) - \ell_i(\y,\x_{-i})=
\phi(\x_i,\x_{-i}) - \phi(\y,\x_{-i}).$
To this end, define an auxiliary variable $\z$ that does not depend on $\x_i$ or $\y$:
$$\z := w_0 \x_0 + \sum_{j\neq i} w_j \x_j.$$
Then 
$\ell_i(\x_i,\x_{-i})=\|w_i\x_i + (\z-\t_i)\|^2 = \|w_i\x_i\|^2 + 2w_i\x_i^\top (\z-\t_i) + \|\z-\t_i\|^2$,
and
\begin{align*}
&\ell_i(\x_i,\x_{-i})-\ell_i(\y,\x_{-i})\\
&=\|w_i\x_i\|^2 + 2w_i\x_i^\top(\z-\t_i) - \|w_i\y\|^2 - 2w_i\y^\top(\z-\t_i).
\end{align*}
On the other hand, 
\begin{multline*}
\phi(\x_i,\x_{-i})=\|w_i\x_i + \z\|^2 - 2 w_i \t_i^\top \x_i - 2 \sum_{j\neq i} w_j \t_j^\top \x_j \\
=\|w_i\x_i\|^2 + 2 w_i (\z-\t_i)^\top \x_i + \|\z\|^2 - 2 \sum_{j\neq i} w_j \t_j^\top \x_j. 
\end{multline*}
The last two terms do not depend on $\x_i$.  Hence
\begin{align*}
&\phi(\x_i,\x_{-i})-\phi(\y,\x_{-i}) \\
&= \|w_i\x_i\|^2 + 2 w_i (\z-\t_i)^\top \x_i - \|w_i\y\|^2 - 2 w_i (\z-\t_i)^\top \y \\
&= \ell_i(\x_i,\x_{-i})-\ell_i(\y,\x_{-i}). 
\end{align*}
\end{proof}

We next show that $\x\in\X^n$ is a pure NE of $G$ if and only if $\x$ is a minimum of $\phi$ restricted to the domain $\X^n$. 

\begin{prop} [Pure NEs $\iff$ minima] \label{prop:cvx} 
The set of pure Nash equilibria in $G$ is
\begin{align}
\NEG &= \mathop{\mathrm{argmin}}_{\x\in\X^n} \phi(\x).
\end{align}\end{prop}
\begin{proof} 
Our potential function $\phi$ is the sum of two convex functions in $\x$ and hence convex.
Since the domain $\X^n$ is convex and $\phi$ is smooth and convex, by Theorem 1 of~\cite{neyman1997correlated} and its corollary, the set of pure Nash equilibria coincides with the minima of the potential function on $\X^n$.
\end{proof}

We remark that by definition $\X$ is compact and convex, thus $\X^n$ is bounded and closed.  
The potential function $\phi(\x)$ may not have a global minimum on the extended domain $\R^{nd}$ (it could diverge to $-\infty$ there), but on $\X^n$ it will have at least one minimum (perhaps on the boundary).
In fact, we have the following guarantee.

\begin{cor} [Cardinality of $\NEG$] \label{cor:cvx} 
$G$ has either one pure NE or infinite pure NEs.
\end{cor}
\begin{proof}  
Since the set of minima of the convex potential function on a compact domain $\X^{n}$ is non-empty, there is at least one pure Nash equilibrium.
Since the set of minima of the convex potential function on a convex domain $\X^{n}$ is convex (corollary in~\cite{neyman1997correlated}), any linear combination of two distinct pure Nash equilibria is another pure Nash equilibrium.
\end{proof}


We now provide a few illustrative examples of BIG.

\begin{eg}[Two influencers with 1D actions]
\label{eg:1D}
There are $n=2$ influencers whose individual action space is $\X = \left[a, b\right]$ (we use $a = 0, b = 6$ in Figures~\ref{fig:bdc} and~\ref{fig:itcl}). 
The receiver takes the average: $\hat x = {x_1+x_2\over 2}$.  
If the targets satisfy $t_1<{a+b\over2}$ and $t_2>{a+b\over 2}$ as in Figure~\ref{fig:bdc}, then there is a unique pure NE: $\NEG=\{(x_1=a, x_2=b)\}$.
Note $(x_1=t_1, x_2=t_2)$ is in general not a NE: if $x_2=t_2$, the first influencer will want to take a more extreme action $x_1<t_1$ to drive the receiver $\hat x = {x_1+x_2\over 2}$ closer to its target $t_1$, and vice versa for the second influencer, \textit{ad infinitum} until they hit the boundary.
In other words, \textbf{both influencers must exaggerate their actions to the maximum extent possible.}
In fact, this is the well-known best response dynamics in potential games which we discuss later. 
As a result, both influencers will end up playing at the opposite boundary of $\X$.
No one is entirely happy: the receiver ends up in the middle $\hat x={a+b\over2}$ so no influencer achieves their target.
Still, this is the best each influencer can do under the presence of other influencers.


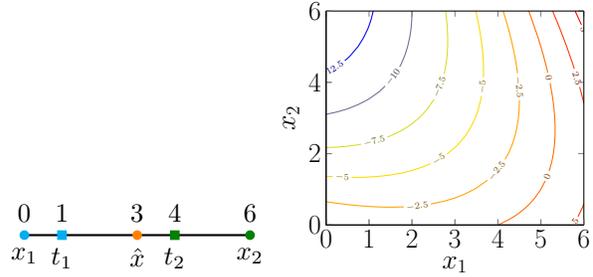
\begin{figure} \centering \begin{tikzpicture} [scale = 0.5] 
\draw[thick] (0.0, 0.0) -- (6.0, 0.0);
\draw[cyan, fill=cyan, thick] (0.9, -0.1) rectangle (1.1, 0.1);
\node[below] at (1.0, -0.1){$t_{1}$};
\draw[black!50!green, fill=black!50!green, thick] (3.9, -0.1) rectangle (4.1, 0.1);
\node[below] at (4.0, -0.1){$t_{2}$};
\draw[cyan, fill=cyan, thick] (0.0, 0.0) circle [radius = 0.1];
\node[below] at (0.0, -0.1){$x_{1}$};
\draw[black!50!green, fill=black!50!green, thick] (6.0, 0.0) circle [radius = 0.1];
\node[below] at (6.0, -0.1){$x_{2}$};
\draw[orange, fill=orange, thick] (3.0, 0.0) circle [radius = 0.1];
\node[below] at (3.0, -0.1){$\hat{x}$};
\node[above] at (0.0, 0.1){$0$};
\node[above] at (6.0, 0.1){$6$};
\node[above] at (3.0, 0.1){$3$};
\node[above] at (1.0, 0.1){$1$};
\node[above] at (4.0, 0.1){$4$};
\end{tikzpicture}
\includegraphics[scale=0.5]{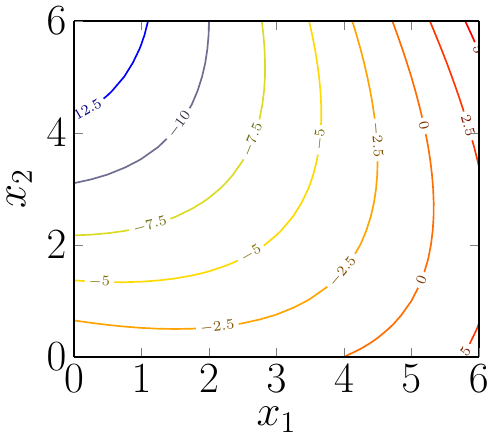}
\caption{(left) Both players maximally exaggerate their actions with pure NE $(x_1=0, x_2=6)$. (right) the potential function $\phi$}
    \label{fig:bdc} \end{figure}

A more nuanced situation happens if $t_1, t_2$ are on the same side of ${a+b\over 2}$, for example the left side in Figure~\ref{fig:itcl}.
There is still a unique pure NE; at the NE both influencers still need to misrepresent their target.
However, the influencer whose target is closer to the center point (in this example, $t_2$) can claim victory: 
The other influencer simply runs out of more left-leaning actions and has to stop at the left boundary $x_1=a$.
The winning influencer best-responds with $x_2=2 t_2 - a$, so that the receiver will end up at its target $\hat x = t_2$.
Thus $\NEG=\{(x_1 = a, x_2 = 2 t_2 - a)\}$ in this example.
Interestingly, if in addition $t_2 > \left(a + {a+b\over 2}\right)/2$ then influencer 2 indeed has a left-leaning target but has to misrepresent itself as right-leaning ($x_2>{a+b\over 2}$) to the receiver.
We will generalize this example in Theorem~\ref{thm:bdy}, where we show at most one influencer can be interior.


\begin{figure} \centering \begin{tikzpicture} [scale = 0.5] 
\draw[thick] (0.0, 0.0) -- (6.0, 0.0);
\draw[cyan, fill=cyan, thick] (0.9, -0.1) rectangle (1.1, 0.1);
\node[below] at (1.0, -0.1){$t_{1}$};
\draw[black!50!green, fill=black!50!green, thick] (1.9, -0.1) rectangle (2.1, 0.1);
\node[below] at (2.0, -0.1){$t_{2}$};
\draw[cyan, fill=cyan, thick] (0.0, 0.0) circle [radius = 0.1];
\node[below] at (0.0, -0.1){$x_{1}$};
\draw[black!50!green, fill=black!50!green, thick] (4.0, 0.0) circle [radius = 0.1];
\node[below] at (4.0, -0.1){$x_{2}$};
\draw[orange, fill=orange, thick] (2.0, 0.0) circle [radius = 0.075];
\node[above] at (2.0, 0.1){$2$};
\node[above] at (2.0, 1.1){$\hat{x}$};
\draw[black, fill=black, thick] (3.0, 0.0) circle [radius = 0.03];
\node[above] at (0.0, 0.1){$0$};
\node[above] at (6.0, 0.1){$6$};
\node[above] at (3.0, 0.1){$3$};
\node[above] at (1.0, 0.1){$1$};
\node[above] at (4.0, 0.1){$4$};
\end{tikzpicture}
\includegraphics[scale=0.5]{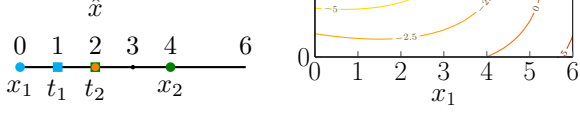}
\caption{Pure NE $(x_1=0, x_2=4)$. Influencer 2 not at boundary.} \label{fig:itcl} \end{figure}

\end{eg}

\begin{eg} [2D actions] \label{eg:tdpt} 
Same as Example~\ref{eg:1D} but let $d = 2$ and $\X = \left[-a, a\right] \times \left[-a, a\right]$.
The game may now have an infinite number of pure NEs.
For example, let the targets be
$\t_1 = \begin{bmatrix} 1 \\ 0 \end{bmatrix}$ and $\t_2=-\t_1$ as in Figure~\ref{fig:fc} (left).
Then $\NEG=\left\{\left(\begin{bmatrix} -a \\ -z \end{bmatrix} , \begin{bmatrix} a \\ z \end{bmatrix} \right): z \in \left[-a, a\right] \right\}$.
These infinite many pure NEs are indicated by the blue and green line segments, paired through the origin. 
All of them are exaggerations from both influencers in terms of the $x$-axis.
All of them result in the receiver arriving at the origin.

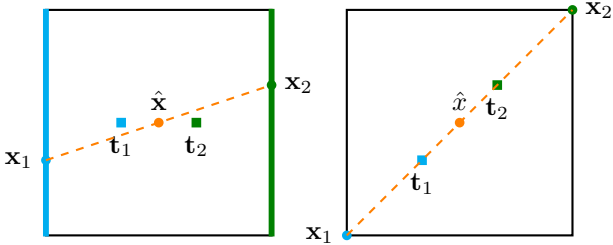
\begin{figure}[H] \centering \begin{tikzpicture} [scale = 0.5] 
\draw[thick] (0.0, 0.0) rectangle (6.0, 6.0);
\draw[cyan, fill=cyan, thick] (0.0, 2.0) circle [radius = 0.1];
\node[left] at (-0.1, 2.0){$\x_{1}$};
\draw[black!50!green, fill=black!50!green, thick] (6.0, 4.0) circle [radius = 0.1];
\node[right] at (6.1, 4.0){$\x_{2}$};
\draw[orange, fill=orange, thick] (3.0, 3.0) circle [radius = 0.1];
\node[above] at (3.0, 3.1){$\hat{\x}$};
\draw[cyan, fill=cyan, thick] (1.9, 2.9) rectangle (2.1, 3.1);
\node[below] at (2.0, 2.9){$\t_{1}$};
\draw[black!50!green, fill=black!50!green, thick] (3.9, 2.9) rectangle (4.1, 3.1);
\node[below] at (4.0, 2.9){$\t_{2}$};
\draw[cyan, fill=cyan, thick] (-0.05, 0.0) rectangle (0.05, 6.0);
\draw[black!50!green, fill=black!50!green, thick] (5.95, 0.0) rectangle (6.05, 6.0);
\draw[dashed, orange, thick] (0.0, 2.0) -- (6.0, 4.0);
\draw[thick] (8.0, 0.0) rectangle (14.0, 6.0);
\draw[cyan, fill=cyan, thick] (8.0, 0.0) circle [radius = 0.1];
\node[left] at (7.9, 0.0){$\x_{1}$};
\draw[black!50!green, fill=black!50!green, thick] (14.0, 6.0) circle [radius = 0.1];
\node[right] at (14.1, 6.0){$\x_{2}$};
\draw[orange, fill=orange, thick] (11.0, 3.0) circle [radius = 0.1];
\node[above] at (11.0, 3.1){$\hat{x}$};
\draw[cyan, fill=cyan, thick] (9.9, 1.9) rectangle (10.1, 2.1);
\node[below] at (10.0, 1.9){$\t_{1}$};
\draw[black!50!green, fill=black!50!green, thick] (11.9, 3.9) rectangle (12.1, 4.1);
\node[below] at (12.0, 3.9){$\t_{2}$};
\draw[dashed, orange, thick] (8.0, 0.0) -- (14.0, 6.0);
\end{tikzpicture} \captionof{figure}{Examples of infinite (left) and unique (right) pure Nash equilibria in $d=2$}\label{fig:fc} \end{figure}

In contrast, if the targets are arranged as in Figure~\ref{fig:fc} (right), there will only be a unique pure NE: $\NEG=\left\{
\left(\begin{bmatrix} -a \\ -a \end{bmatrix} , \begin{bmatrix} a \\ a \end{bmatrix} \right)
\right\}$.
\end{eg}

As seen from these examples, in BIG the pure NEs often involve all influencers misrepresenting their true target to the receiver (i.e. $\x_i \neq \t_i$).
Furthermore, such misrepresentation often takes the form of \emph{extreme exaggeration}, in the sense that an influencer's rational action $\x_i$ at any pure NE is often pushed to the boundary of the action space $\X$ so they cannot exaggerate the action further.
Our next theorem precisely quantifies this phenomenon. We prove that, provided that the influencers' targets $\t_1 \ldots \t_n$ are all distinct, \textbf{at any pure NE extreme exaggeration is necessary for all but at most one influencer}.
It is possible that all influencers must perform extreme exaggeration.
Furthermore, if there exists one influencer (say influencer $i^*$) who does not, then it is the winner in that the receiver will end up at its target $\t_{i^*}$; Still, this winning influencer in general also need to misrepresent its target $\x_{i^*} \neq \t_{i^*}$ to the receiver, it is just that $\x_{i^*}$ is in the interior of $\X$ and not extreme.


\begin{thm} [All-But-At-Most-One Extreme Exaggeration] \label{thm:bdy} 
If $\left\{\mathbf{t}_{i}\right\}_{i=1}^{n}$ are all distinct, then every pure NE $\left(\mathbf{x}_{1}, \mathbf{x}_{2}, ..., \mathbf{x}_{n}\right)$ satisfies the property that,
\begin{align}
\left| \left\{i\in[n] : \mathbf{x}_{i} \in \mathrm{int\;} \X\right\} \right| &\leq 1.
\end{align}
Furthermore, if $\x_{i^\star} \in \mathrm{int\;} \X$ for some $i^\star \in \left[n\right]$, then,
\begin{align}
\hat{\mathbf{x}} &= \mathbf{t}_{i^\star}.
\label{eq:ontarget}
\end{align}\end{thm}
\begin{proof}  \label{proof:bdypf} 
For any $i \in [n]$, 
\begin{align}
\nabla_{\x_i} \phi = 2 w_i \sum_{k=0}^n w_k \x_k - 2 w_i \t_i.
\end{align}
Suppose $\exists i, j\in [n], i\neq j: \x_i, \x_j \in \mathrm{int\;} \X$. Then 
\begin{align}
& \nabla_{\x_i} \phi = \boldsymbol{0} = \nabla_{\x_j} \phi \\
\Rightarrow & \t_i =  \sum_{k=0}^n w_k \x_k = \t_j,
\end{align}
a contradiction.
Equation~\eqref{eq:ontarget} follows from $\nabla_{\x_{i^*}} \phi = \boldsymbol{0}$ and the definition of the receiver~\eqref{eq:receiver}.
\end{proof}

Finally, we remark on algorithms to find pure NEs for BIG.
Due to Proposition~\ref{prop:cvx}, any convex optimization algorithm that minimizes the convex function $\phi$ over the convex set $\X^n$ with strong guarantees can be utilized to find a pure NE~\cite{boyd2004convex}.
Meanwhile, in the game theory community the best-response dynamics is a traditional algorithm for finding a pure NE in potential games: ~\cite{roughgarden2010algorithmic} 
\begin{df} [Best Response Dynamics] \label{prop:br} 
Starting from an arbitrary $\x^{\left(0\right)} \in \X^n$, the best response sequence $\left\{\x^{\left(t\right)}\right\}_{t=1}^{\infty}$ converges to a pure Nash equilibrium of $G$, where,
\begin{align}
\mathbf{x}^{\left(t\right)}_{i} &= \begin{cases} \mathop{\mathrm{argmin}}_{\mathbf{x}_{i}} \phi\left(\mathbf{x}_{i}, \mathbf{x}^{\left(t - 1\right)}_{-i}\right) & \text{\;if\;} i = t \mod n \\ \mathbf{x}^{\left(t-1\right)}_{i} & \text{\;otherwise.\;} \\ \end{cases} 
\end{align}
\end{df}

Of course, best response dynamics correspond to coordinate descent on $\phi$.
One interesting observation that is relevant for BIG is that, under best response dynamics, no influencer needs to know other influencers' targets.
This removes the requirement that $\t_1, \ldots, \t_n$ must be common knowledge.
Concretely, the influencers may carry out the best response dynamics as a learning dynamics in a distributed fashion over time, with no two influencers simultaneously updating their actions.
When influencer $i$ updates its own action $\x_i$, it observes other player's most recent actions $\x_{-i}$ but does not need to know their targets $\t_{-i}$.
This is because minimizing the potential function $\phi$ along the $\x_i$ direction is equivalent to minimizing its own loss function $\ell_i$:
\begin{align}
\mathop{\mathrm{argmin}}_{\x_i} \phi(\x_i, \x_{-i})=
\mathop{\mathrm{argmin}}_{\x_i} \ell_{i}(\x_i, \x_{-i}).
\end{align}
Therefore, the best response dynamics may offer a computational account on how influencers in the real world iteratively adjust their actions based on actions of other influencers without knowing the other influencers' true intentions, and still reaching an equilibrium.

\section{Extensions of BIG}

\subsection{An Alternative Player Loss Function}
\label{sec:cosine}
Up to now influencer $i$'s loss function~\eqref{eq:ell} is based on the Euclidean distance between its target point $\t_i \in \X$ and the receiver $\hat\x$.
In some applications, the following \emph{negative inner product loss} can be more appropriate:
\begin{equation}
\label{eq:ell2}
\ell_i(\x) := - \t_i^\top \hat\x.
\end{equation}
That is, influencer $i$ has a small loss if the receiver $\hat\x$ has a large projection onto the direction of target direction $\t_i$.

BIG with this negative inner product loss~\eqref{eq:ell2} has an even stronger guarantee: the game has a \emph{Weakly Dominant Strategy Equilibrium}. 

\begin{df} [Weakly Dominant Strategy Equilibrium (wDSE)] \label{df:sdse} 
An action profile $\x=(\x_1 \ldots \x_n) \in \X^{n}$ is a wDSE if for every player $i$,
\begin{align}
\label{eq:wDSE}
\ell_{i}\left(\x_{i}, \x_{-i}\right) &\leq l_{i}\left(\y, \x_{-i}\right), \; \forall \y \in \X, \x_{-i} \in \X^{n-1}.
\end{align}

\end{df}

\begin{rmk}  \label{rmk:ea} 
The term weakly dominant strategy equilibrium is used in Chapter 4.5 of~\cite{tadelis2013game}, and it is also called dominant strategy equilibrium in Chapter 10.3 of~\cite{osborne1994course}, and dominant strategy solution in Chapter 1.3.1 of~\cite{roughgarden2010algorithmic}. As noted in~\cite{osborne1994course}, an action in a wDSE is not required to weakly dominate all other actions for a player since the player could have multiple actions that are equivalent, all of which dominate the remaining actions. wDSE is also a weaker solution concept compared to (strictly) dominant strategy equilibrium, which requires strict inequality everywhere,
\begin{align}
\ell_{i}\left(\x_{i}, \x_{-i}\right) < l_{i}\left(\y, \x_{-i}\right), \; \forall \y \in \X, \x_{-i} \in \X^{n-1}.
\end{align}
\end{rmk}

\begin{thm}[Existence of wDSE] 
Under loss~\eqref{eq:ell2}, game $G$ has a wDSE $\x_i^*$ satisfying,
\begin{equation}
\label{eq:DSE}
\x_i^* \in \argmax_{\x_i \in \X} w_i \t_i^\top \x_i
\;,\forall i\in[n].
\end{equation}
\end{thm}

\begin{proof}
Given an arbitrary joint action from other players $\x_{-i}$, player $i$'s best response is
\begin{align}
BR(\x_{-i}) &:= \argmin_{\x_i \in \X} \ell_i(\x_i, \x_{-i}) \\
            &= \argmin_{\x_i \in \X} - \t_i^\top \hat\x \\
 &= \argmin_{\x_i \in \X} - \t_i^\top \left( w_0 \x_0 + \sum_{j=1}^n w_j \x_j \right) \\
 &= \argmin_{\x_i \in \X} - w_i \t_i^\top \x_i + \mathrm{const} = \x_i^*.
\end{align}
The best response is independent of $\x_{-i}$.
\end{proof}

One significant benefit of this wDSE is that each influencer can compute their own $\x_i^*$ without the knowledge of other influencers' weights $w_j$ and target $\t_j$, $\forall j\neq i$.
This removes the requirement that $\x_0, \{w_{0:n}\}, \{\t_{1:n}\}$ are common knowledge from Definition~\ref{df:game}.
The wDSE action $\x_i^*$ in~\eqref{eq:DSE} is determined in part by the sign of $w_i$. 
When $w_i<0$ this is akin to reverse psychology: knowing that the receiver will flip its action direction in~\eqref{eq:receiver}, influencer $i$ should intentionally go against its own target.
We remark that $\x_i^*$ may not be exactly along the direction of $w_i \t_i$ since it depends on the domain $\X$.
Nonetheless, computing $\x_i^*$ is a convex optimization problem--maximizing a linear function over a convex set--and thus efficient.

\begin{eg} [Unique wDSE] \label{eg:udse} 
In this example, $d = 2$ and $\X$ is given by the diamond shape in Figure~\ref{fig:cos} (left), with $\t_{1} = \begin{bmatrix} 1 \\ 0 \end{bmatrix}$ and $\t_{2} = \begin{bmatrix} 0 \\ 1 \end{bmatrix}$ . The unique wDSE is the action profile containing the extreme points in $\X$ in the direction of $\t_{1}$ and $\t_{2}$ .

\end{eg}

\begin{eg} [Infinite wDSE] \label{eg:idse} 
In this example, $d = 2$ and $\X = \left[-a, a\right] \times \left[-a, a\right]$, with $\t_{1} = \begin{bmatrix} 1 \\ 0 \end{bmatrix}$ and $\t_{2} = \begin{bmatrix} 0 \\ 1 \end{bmatrix}$ . All actions for player $1$ along the blue line are equivalent and lead to the same loss regardless of player $2$'s action, and all other actions in $\X$ except for these actions along the blue line are strictly dominated. Similarly, all actions for player $2$ along the green line are equivalent and strictly dominate all other actions. As a result, given Definition~\ref{df:sdse}, any pair of actions, one on the blue line for player $1$ and one on the green line for player $2$, is a wDSE.  There are infinite number of them.

\end{eg}

\begin{figure}[H] \centering \begin{tikzpicture} [scale = 0.5]
\draw[thick] (0.0, 6.0) -- (2.0, 2.0) -- (6.0, 0.0) -- (4.0, 4.0) -- (0.0, 6.0);
\draw[->, cyan, ultra thick] (3.0, 3.0) -- (4.0, 3.0);
\node[below right] at (4.0, 3.0){$\t_{1}$};
\draw[->, black!50!green, ultra thick] (3.0, 3.0) -- (3.0, 4.0);
\node[above left] at (3.0, 4.0){$\t_{2}$};
\draw[cyan, fill=cyan, thick] (6.0, 0.0) circle [radius = 0.1];
\node[right] at (6.0, 0.0){$\x_{1}$};
\draw[black!50!green, fill=black!50!green, thick] (0.0, 6.0) circle [radius = 0.1];
\node[above] at (0.0, 6.0){$\x_{2}$};
\draw[cyan, dashed, thick] (6.0, 0.0) -- (6.0, 6.0);
\draw[black!50!green, dashed, thick] (0.0, 6.0) -- (6.0, 6.0);
\draw[orange, fill=orange, thick] (3.0, 3.0) circle [radius = 0.1];
\node[above right] at (3.0, 3.0){$\hat{\x}$};
\draw[orange, dashed, thick] (6.0, 0.0) -- (0.0, 6.0);
\draw[thick] (8.0, 0.0) rectangle (14.0, 6.0);
\draw[cyan, ultra thick] (14.0, 0.0) -- (14.0, 6.0);
\draw[black!50!green, ultra thick] (8.0, 6.0) -- (14.0, 6.0);
\draw[->, cyan, ultra thick] (11.0, 3.0) -- (12.0, 3.0);
\node[right] at (12.0, 3.0){$\t_{1}$};
\draw[->, black!50!green, ultra thick] (11.0, 3.0) -- (11.0, 4.0);
\node[above] at (11.0, 4.0){$\t_{2}$};
\draw[cyan, fill=cyan, thick] (14.0, 3.0) circle [radius = 0.1];
\node[right] at (14.0, 3.0){$\x_{1}$};
\draw[black!50!green, fill=black!50!green, thick] (11.0, 6.0) circle [radius = 0.1];
\node[above] at (11.0, 6.0){$\x_{2}$};
\draw[orange, fill=orange, thick] (14.0, 6.0) circle [radius = 0.1];
\node[above] at (14.0, 6.0){$\x'_{1} = \x'_{2} = \hat{\x'}$};
\draw[orange, fill=orange, thick] (12.5, 4.5) circle [radius = 0.1];
\node[above right] at (12.5, 4.5){$\hat{\x}$};
\draw[orange, dashed, thick] (14.0, 3.0) -- (11.0, 6.0);
\end{tikzpicture} \captionof{figure}{Examples of unique wDSE (left) and infinite number of wDSEs (right) in $d = 2$}\label{fig:cos} \end{figure}
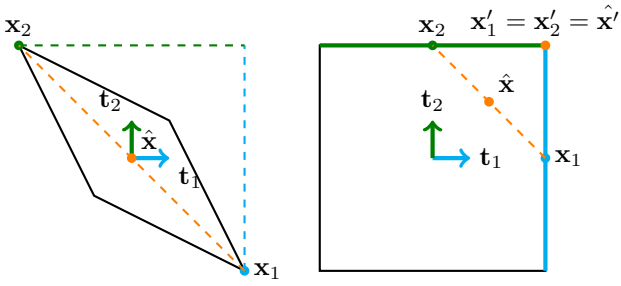 

\subsection{Finite Action Space} 
So far, we have assumed that the influencers' action space $\X$ is a compact and convex (hence infinite unless singleton) subset of $\R^d$. 
In some applications, the influencers are restricted to picking their actions from a finite $\X$ instead.
For example, $\X$ may be the collection of news articles published by all professional news agencies within the past 24 hours, and each influencer may select a handful of such news articles to place on a social media user (the receiver)'s timeline.
This motivates the extension to finite action space:

\begin{df} [BIG with finite action space] \label{df:dbig} 
Battling Influencer Game with finite action space is an $n$-player general-sum game $F = \left(n, \left\{\mathcal{D}^{\left(k_i\right)}\right\}_{i=1}^{n}, \left\{\ell_i\right\}_{i=1}^{n}\right)$, where the action space of player $i, \mathcal{D}^{\left(k_{i}\right)}$ is the set of all subsets containing $k_{i}$ elements (optionally allow repeats) from a finite $\X \subset \mathbb{R}^{d}$.
The loss function of player $i$
is given by $\ell_{i}\left(\x\right) = \left\|\hat{\x} - \t_{i}\right\|^{2}_{2}$ with
\begin{align}
\hat{\x} &= w_{0} \x_{0} + \displaystyle\sum_{i=1}^{n} w_{i} \displaystyle\sum_{j=1}^{k_{i}} \x^{\left(j\right)}_{i} , 
\end{align}
where $\x^{(j)}_i$ is the $j$th element in player $i$'s chosen subset.
The parameters $\x_{0}, \left\{w_i\right\}_{i=0}^{n}, \left\{\t_i\right\}_{i=1}^{n}$ and $\left\{k_i\right\}_{i=0}^{n}$ are common knowledge to all players.
\end{df}

In the original BIG $G$ (Definition~\ref{df:game}) where $\X$ is convex, allowing the players to choose multiple items would not affect the results since choosing multiple items is equivalent to choosing the mean of these items, which is still in $\X$. In the new game, the average of items in $\X$ may not be in $\X$.
However, the new 
game can be interpreted as each influencer picks one ``meta item'' $\x'_i$ instead of $k_i$ items, 
with 
$w'_{i} = w_{i} k_{i}$,
$\x'_{i}= \dfrac{1}{k_{i}} \displaystyle\sum_{j=1}^{k_{i}} \x^{\left(j\right)}_{i}$. 

\begin{prop}
\label{prop:Fpotential}
$F$ is also a potential game, with a potential function
$\phi\left(\left\{\left\{\x^{\left(j\right)}_{i}\right\}_{j=1}^{k_{i}}\right\}_{i=1}^{n}\right):=$
\begin{align}
\left\|w_{0} \x_{0} + \displaystyle\sum_{i=1}^{n} w'_{i} \x'_{i}\right\|_{2}^{2} - 2 \displaystyle\sum_{i=1}^{n} w'_{i} \t^\top_{i} \x'_{i}.
\end{align}
\end{prop}
Therefore, the new game $F$ also has at least one pure NE.
Since $\X$ is finite and the number of players is finite, $F$ is a finite game.
A pure NE can be found through best response dynamics~\cite{roughgarden2010algorithmic}. 
Unlike the continuous case, each iteration of best response dynamics can be costly to compute for large values of $k_{i}$ since it involves solving a variant of the subset sum problem. 


Proposition~\ref{prop:Fpotential} shows that the new game $F$ must have at least one pure NE.
We now show a non-trivial example where $F$ can have an exponential number of pure NEs even when there are only $n=2$ players.
This is true even if $\X$ contains distinct elements, and the influencers cannot repeat chosen elements.

\begin{eg} [Many pure NEs in $F$] \label{eg:expne} 
Consider an instance of game $F$ with
$n = 2$, $d = 1$, $t_{1} = - {1\over 4}, t_{2} = {1\over 4}$, $k_{1} = k_{2} = k$,
the receiver takes simple average $\hat x = {1\over 2k} \left( \sum_{j=1}^k x_1^{(j)}+\sum_{j=1}^k x_2^{(j)} \right)$.
Let $\X = \X_{-} \cup \X_{+}$ where
\begin{align}
\X_{-} \coloneqq \left\{-2 k \cdot 2^{i}\right\}_{i=0}^{|\X|/2-1},
\X_{+} &\coloneqq \left\{2 k \cdot 2^{i}\right\}_{i=0}^{|\X|/2-1}.
\end{align}
Assume $|\X|\ge 2k$.
Consider any player 2 action $\x_2$ which is a subset of $\X_+$ of size $k$.
Note ${1\over 2k}\sum_{j=1}^k x_2^{(j)}$ is a positive integer with $\x_2$ indexing its binary representation.
For example, if $k=2$ and $\x_2=\{2k\cdot 2^0, 2k\cdot 2^1\}$ then this integer is 3.
What is player 1's best response to player 2?
Given player 1's target $t_1=-{1\over 4}$, player 1 should take action $\x_1$ which selects the corresponding negative items in $\X_-$, for example $\x_1=\{-2k\cdot 2^0, -2k\cdot 2^1\}$.
The joint action $\x_1, \x_2$ brings the receiver to $\hat x=0$.
This is the best that player 1 can do: any other $\x'_1$ (i.e. a subset of size $k$ of $\X$) is not a best response because it changes $\hat x$ to a different integer, which is farther away from $t_1$ compared to $\hat x=0$.
Conversely, $\x_2$ is also the best response to that $\x_1$.
Therefore, $(\x_1, \x_2)$ forms a pure NE.
Now, player 2 could have started with $\dbinom{|\X|/2}{k}$ different $\x_2$ subsets, each corresponds to a different pure NE.
Therefore, if we allow $k$ to grow with $|\X|$ such as $k=|\X|/4$, this game has an exponential number of pure NEs.
\end{eg}

\section{Implications of BIG to Value Alignment}

\subsection{Heterogeneous Value Alignment as a Game}

We now show empirically that a stylized version of value alignment can be modeled by BIG.
We are interested in the setting where multiple people with heterogeneous values provide feedback data~\cite{santurkar2023whose,bakker2022fine,chen2024pal}, and that they are aware of the presence of one another. 
Our focus is not on the value alignment algorithm itself, which is fixed and only plays the role of the receiver in BIG.
Instead, we focus on how these people may become strategic in providing their feedback.
Specifically, \textbf{our analysis on BIG implies that rational people will exaggerate their own value stance in anticipation of being ``canceled out'' by one another.}
This may help explain one source of misinformation in social discourse.

Concretely, we connect value alignment and BIG as follows.
Let the $n$ people be the $n$ players in BIG.
We simplify the feedback process by assuming only a single prompt or context which we denote $c$ (to avoid notation conflict with actions $\x$ in BIG).
We adopt the standard ideal point model~\cite{coombs2017psychological,jamieson2011active,singla2016actively,xu2020simultaneous}: 
For this context $c$, an ideal point $\theta \in \R^d$, and a response $\y \in \R^d$ (e.g. a document embedding vector), we define the reward model 
\begin{equation}
r_{\theta}(y) := -\|\y - \theta\|^2
\end{equation}
to measure how good the response $\y$ is for the prompt $c$ according to the ideal point $\theta$.
The closer $\y$ is to $\theta$, the higher the reward.
Given a pair of responses $\y, \y'$, a player with ideal point $\theta$ draws a stochastic binary pairwise judgment label $z \in \{-1,1\}$ according to the Bradley-Terry-Luce (BTL) model~\cite{bradley1952rank}:
\begin{align}
\label{eq:BTL}
P(z \mid \y,\y',\theta) = {1\over 1+\exp(-z (r_{\theta}(y)-r_{\theta}(y')))}
\end{align}
where $z=1$ means $\y \succeq \y'$ and $z=-1$ means $\y \prec \y'$ to this player. 
Each player $i$ will label $N_i$ tuples of the form $(y, y', z)$ where we assume $y,y'$ are i.i.d. from some response distribution $P_Y$.
Finally, the union of tuples from all players are given to the value alignment algorithm as training data.

Crucially, in BIG the players can misreport their values.
Each player $i \in [n]$ has a \emph{true ideal point} $\t_i\in\R^d$ which represents their true value.
These are their targets in BIG, because they hope the value alignment algorithm ultimately arrives at $\t_i$ as well. 
The players have heterogeneous values if the $\t_1 \ldots \t_n$ are distinct.
If the players were truthful, they would each use $\theta=\t_i$ in~\eqref{eq:BTL} when labeling their tuples.
However, BIG allows each player $i$ to choose a \emph{fake ideal point} $\x_i \in \X \subset \R^d$.
Player $i$ instead uses $\theta=\x_i$ in~\eqref{eq:BTL} to label its $N_i$ tuples.
Note the action space $\X$ is the space of ideal points. An action, namely player $i$ choosing ideal point $\x_i$, is eventually reflected in the $N_i$ tuples (in particular the $z$'s) from that player for value alignment.  

The value alignment algorithm plays the role of the receiver in BIG.
But unlike the affine function~\eqref{eq:receiver}, we adopt the standard maximum likelihood estimate for a global ideal point model parameter $\hat \x$, trained from the union of tuples annotated by all players $(\y,\y',z)_{1:N}$ where $N=\sum_{i=1}^n N_i$:
\begin{equation}
\label{eq:MLE}
\hat \x = \argmax_{\x \in \R^d} \sum_{j=1}^N \log P(z_j \mid \y_j, \y'_j, \theta=\x).
\end{equation}
We note that the training data, being a mixture of BTL, is outside the model family in~\eqref{eq:MLE}; \eqref{eq:MLE} itself seems highly nonlinear and depends implicitly on the response distribution $P_Y$.
Nonetheless, our empirical results indicate that the MLE $\hat \x$ is approximately a linear combination of individual player fake ideal points $\x_1 \ldots \x_n$.
This allows us to make predictions on player strategic behaviors based on earlier analysis on BIG. 
In particular, we predict that given the opportunity the players will evolve their fake ideal points $\{\x_i\}$ similar to best-response dynamics; that they do this to make the value alignment algorithm's global ideal point~\eqref{eq:MLE} closer to their true ideal points $\{\t_i\}$; and that they will end up in a Nash equilibrium where their fake ideal points $\{\x_i\}$ are exaggerations of their true values $\{\t_i\}$.
We next present an experiment to support these predictions.

\subsection{A Value Alignment Experiment}

Let there be $n=2$ players with ideal point action space $\X=[-1,1]$.
Their true ideal points are $t_{1} = -0.1, t_{2} = 0.3$ respectively.
We draw $N_1=10000$ i.i.d. pairs 
from the response distribution 
 for player 1 to label: 
$(y,y') \sim P_Y = \mathrm{uniform}[-10, 10]^2$,
and another $N_2=10000$ i.i.d. pairs for player 2.

In iteration 0, both players start truthfully.
Player 1 starts at its true ideal point $x_1=t_1$ to annotate its preferences on the $N_1$ pairs.
That is, for each $(y,y')$ pair player 1 draws a Bernoulli $\pm 1$-valued label $z$ according to~\eqref{eq:BTL} with $\theta=x_1=t_1$.
We show player 1's annotated dataset $(y, y', z)_{N_1}$ in Figure~\ref{va2}(left).
For visualization purpose, we zoom in to the center region $[-3,3]^2$ and also randomly thinned out the data so that the noisy nature of $z$ is easier to see.
Similarly, player 2 annotates its $N_2$ pairs according to~\eqref{eq:BTL} with $\theta=x_2=t_2$ (Figure~\ref{va2} right).

\begin{figure}[H] \centering 
\includegraphics[width=3.0in]{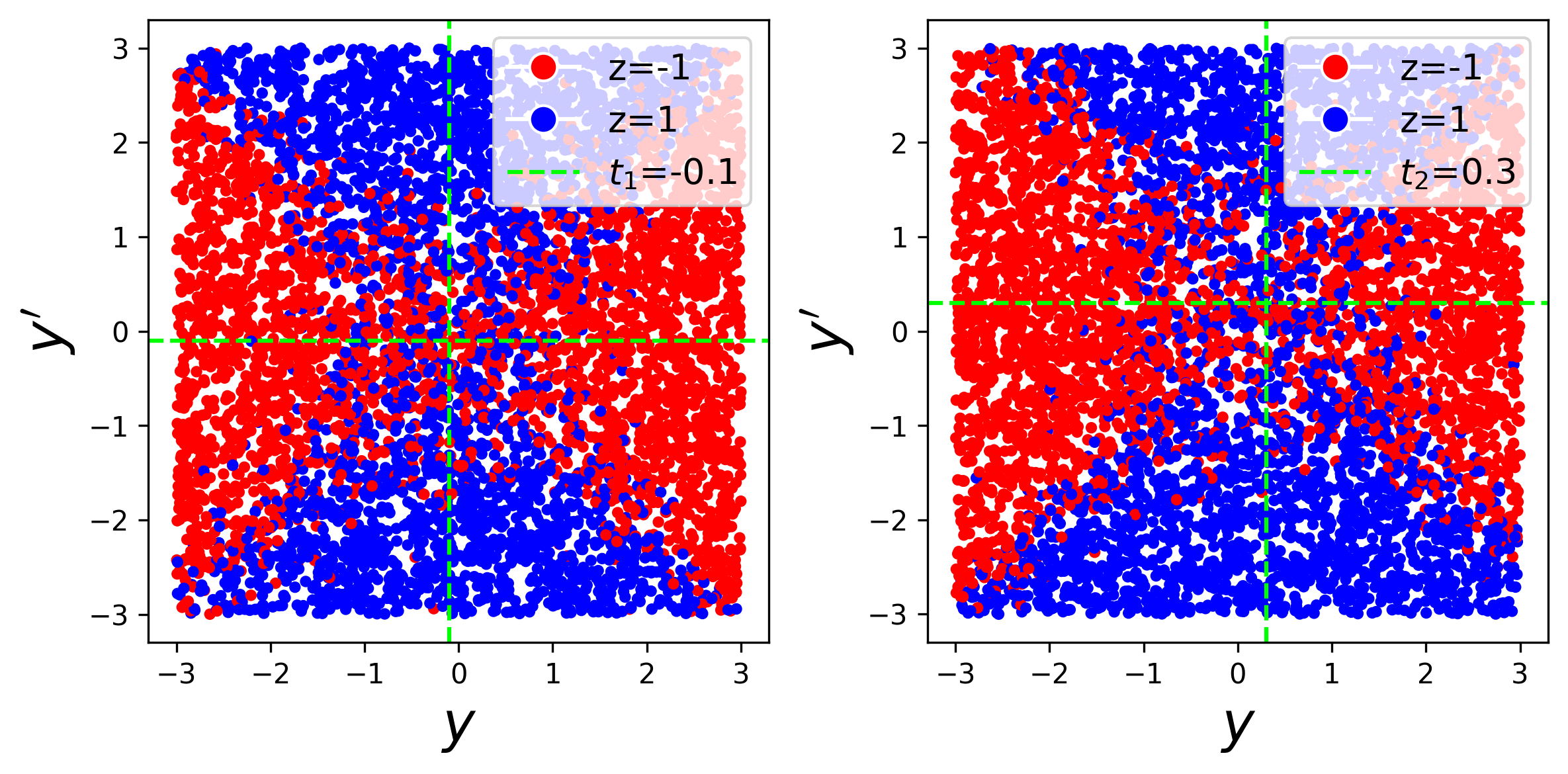}
    \caption{Pairwise preference labels $z$ when both players are truthful. Left: player 1 with $x_1=t_1$, right: player 2 with $x_2=t_2$.}
    \label{va2}
\end{figure}

These $N_1+N_2$ preference tuples are given to the value alignment algorithm (the receiver).
The receiver numerically solves for the MLE by~\eqref{eq:MLE}.  We show the log likelihood surface in Figure~\ref{va1}(right).
The MLE is at $\hat x=0.103$, shown in the same figure(left). 
We observe that the receiver, despite maximizing the log likelihood, can be well-approximated by an affine receiver
\begin{equation}
\label{eq:approxreceiver}
\hat x = {1\over 2} x_1 + {1\over 2} x_2.
\end{equation}
This concludes iteration 0.

\begin{figure}[H] \centering 
\includegraphics[width=3.2in]{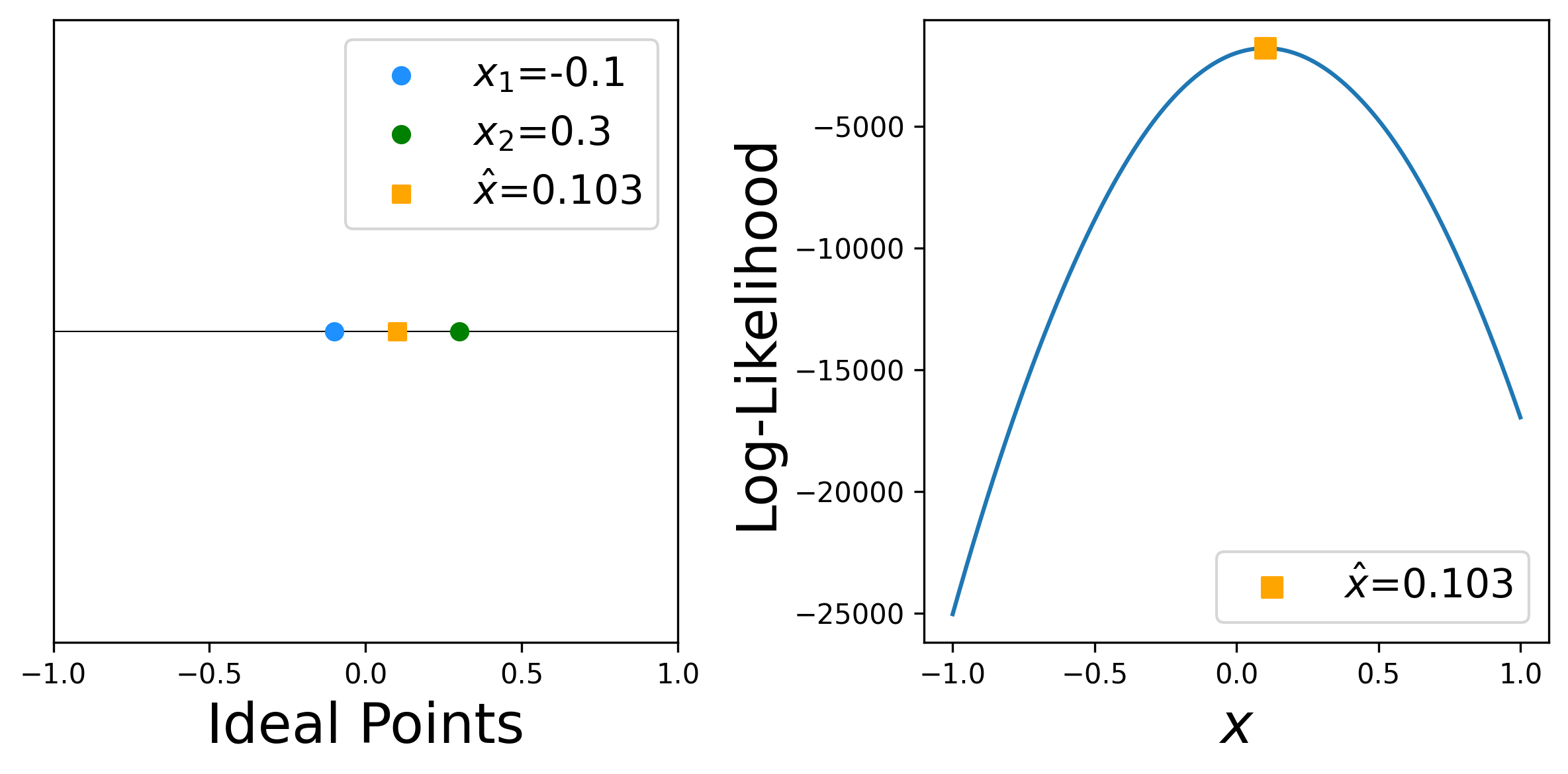}
    \caption{Receiver's MLE $\hat x$ under truthful players.}
    \label{va1}
\end{figure}

In iteration 1, imagine player 1 observes all training data and the value alignment algorithm output from iteration 0 (i.e. a global ideal point at $\hat x=0.103$).
It realizes that the output is far from its true ideal point $t_1=-0.1$.
For the sake of exposition, we now allow player 1 to re-label its $N_1$ tuples using a different (fake) ideal point $x_1$.  Player 2's data remain fixed.  Then we will run value alignment algorithm again.   If player 1 can simulate the value alignment algorithm, it can perform a binary search in $x_1$ to best-respond to player 2, with the goal to move value alignment algorithm output to $t_1$.  We show this binary search in Figure~\ref{va3}.  After 6 binary search steps player 1 finds that $x_1=-0.536$ is good: together with player 2's data this indeed moves value alignment algorithm output to $\hat x = -0.101$, very close to player 1 target $t_1=-0.1$.
This is one iteration of empirical best-response by player 1.
Again, the empirical based-response $x_1=-0.536$ is close to the theoretical best-response under the affine receiver~\eqref{eq:approxreceiver}, which is $x_1=2 t_1 - x_2=-0.5$.

\begin{figure}[H] \centering 
\includegraphics[width=1.8in]{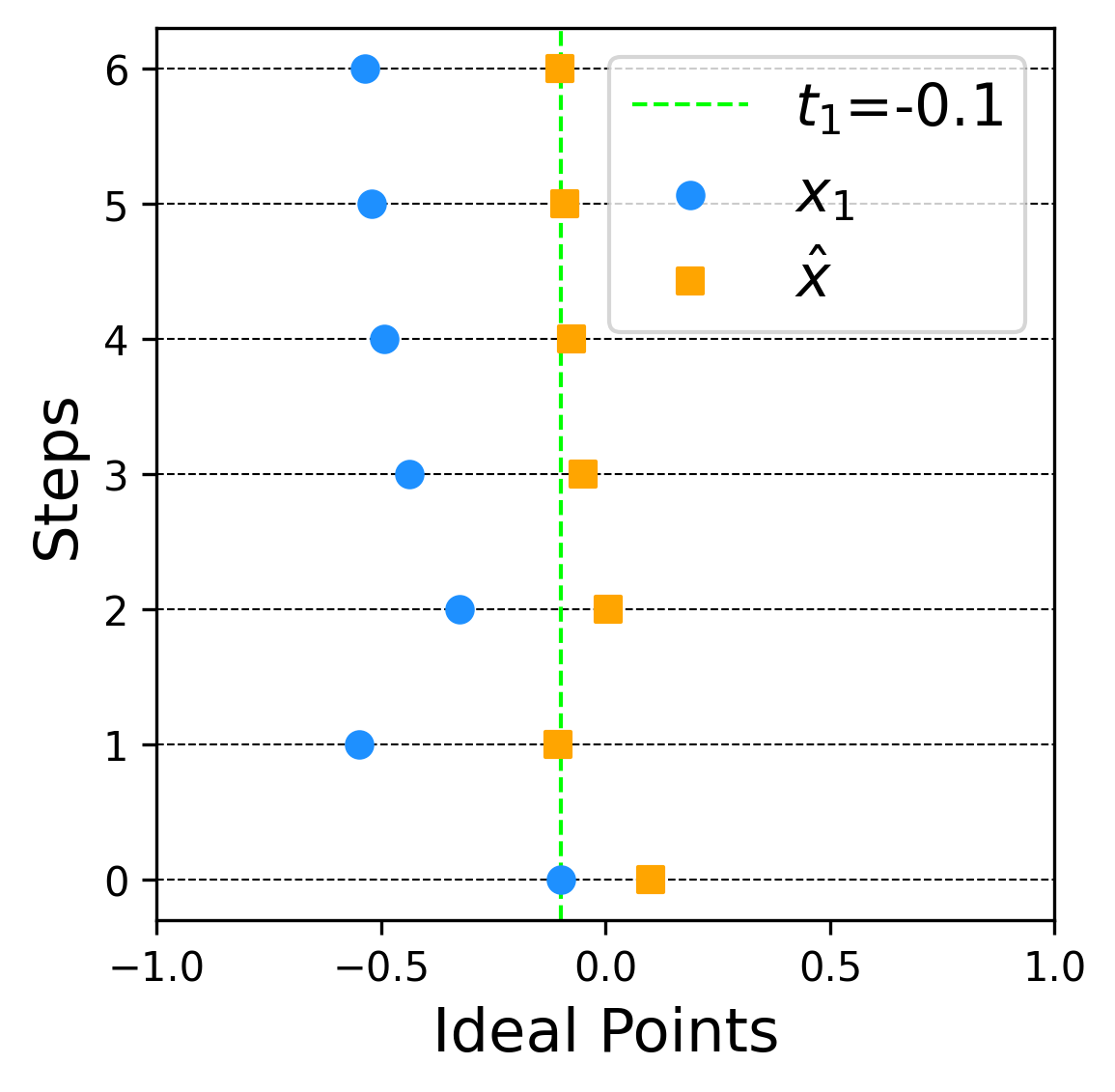}
    \caption{Player 1 conducts a binary search to find the empirical best response $x_1$.}
    \label{va3}
\end{figure}


In subsequent iterations, we allow alternating players to perform such empirical best-response.
Figure~\ref{va4} shows the dynamics.
In iteration 2, player 2 is able to somewhat drag value alignment output $\hat x$ back towards its target.  However, player 2 is limited by the action space: even with the maximum fake ideal point $x_2=1$ the output only moves back to $\hat x=0.209$, not enough to reach its target $t_2=0.3$.
In iteration 3, player 1 fights back with the minimum fake ideal point $x_1=-1$, dragging value alignment output to $\hat x = 0.024$, closer but not reaching its target $t_1=-0.1$.
This is when the dynamics converges to a Nash equilibrium, where neither player can make further improvements. 
Observe at the NE $(x_1=-1, x_2=1)$ the player's fake ideal points are much exaggerated versions of their true ideal points $t_1=-0.1, t_2=0.3$, respectively.

\begin{figure}[htb] \centering 
\includegraphics[width=1.8in]{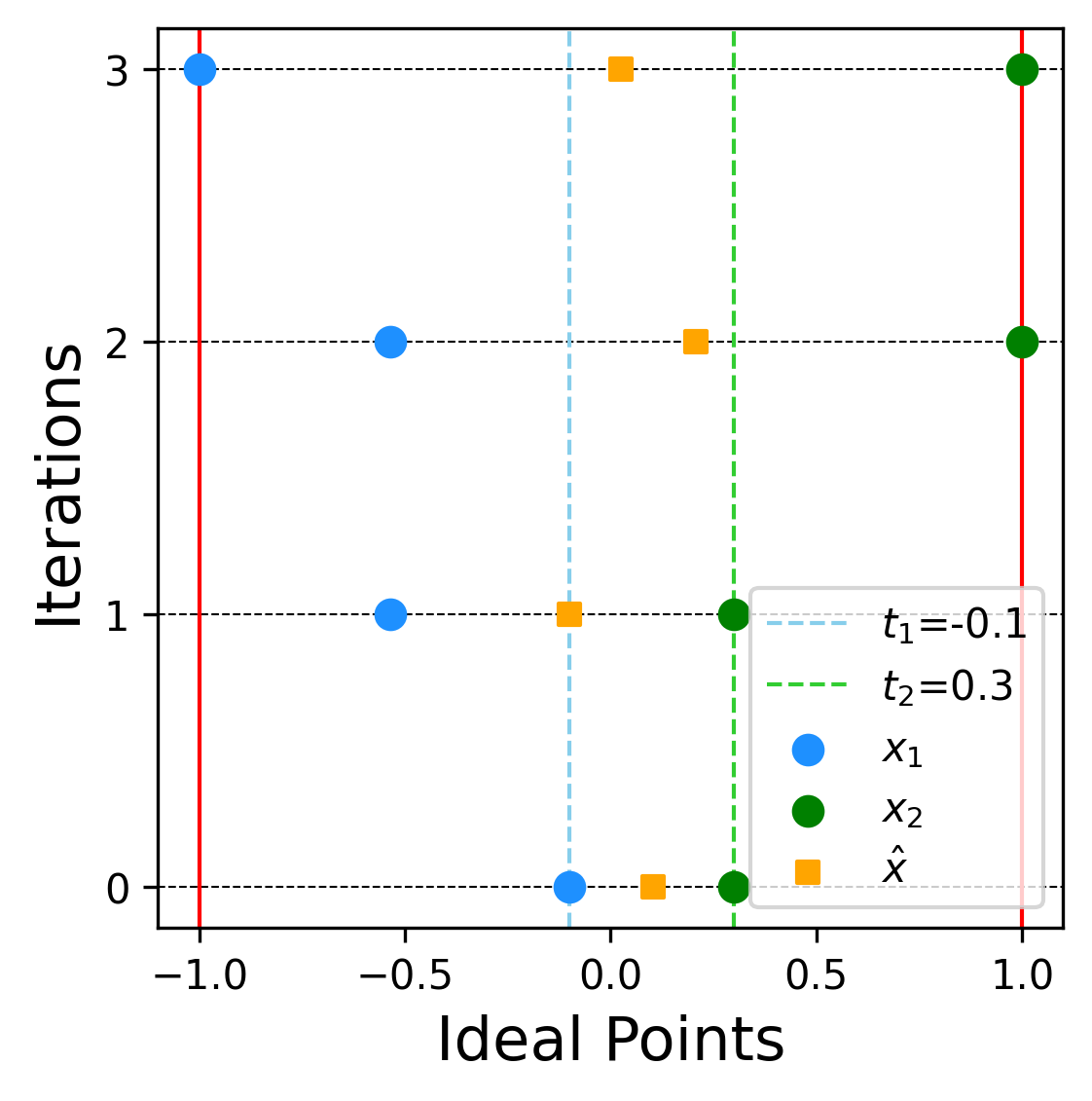}
    \caption{Empirical best-response dynamics converges to an exaggerating Nash equilibrium in iteration 3.}
    \label{va4}
\end{figure}

In Figure~\ref{va5+6} we provide a visualization of the final preference labels $z$ generated by the players after reaching the Nash equilibrium $(x_1 = -1, x_2=1)$. 
This figure is to be contrasted with Figure~\ref{va2}.
Now both players are untruthful and produce preference labels according to their fake, exaggerated ideal points.
This shift illustrates the effect of BIG.
\begin{figure}[H]
    \centering
     \includegraphics[width=3.0in]{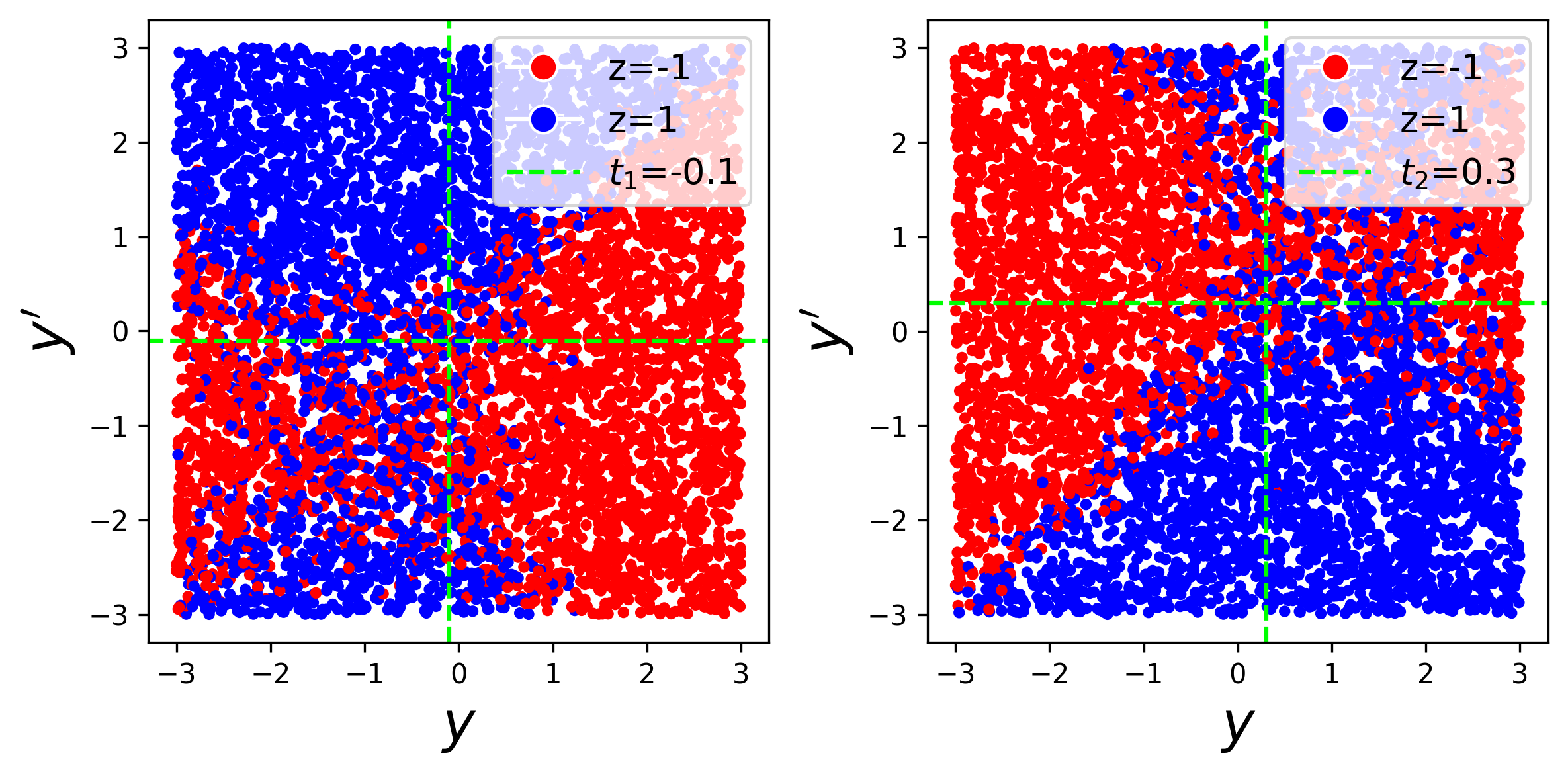}
    \caption{Preference labels $z$ at the Nash equilibrium showing untruthfulness.  Left: player 1 exaggerates with $x_1=-1$; Right: player 2 exaggerates with $x_2=1$.}
    \label{va5+6}
\end{figure}

\section{Conclusion and Future Work}
We proved that a version of the battling influencers game is a potential game, and characterized its pure Nash equilibrium structures.
As a use case, our game applies to standard value alignment via learning from preference feedback.
Consequently, we rationalized a strategic behavior (exaggeration) in alignment data providers.
Future work will focus on mechanism design to remove incentives for such strategic behaviors.


\textbf{Acknowledgment} This project was supported in part by NSF grants 1836978, 2023239, 2202457, 2331669, ARO MURI W911NF2110317, and AF CoE FA9550-18-1-0166.



\bibliography{va}
\bibliographystyle{icml2025}

\end{document}